\documentclass[11pt]{article}
\usepackage{amssymb,amsthm}
\usepackage{mathptmx}
\usepackage{a4}

\newtheorem{theorem}{Theorem}
\newtheorem{thm}{Thm}[section]
\newtheorem{proposition}[thm]{Proposition}

\theoremstyle{definition}
\newtheorem{remark}[thm]{Remark}

\newtheorem{definition}[thm]{Definition}
\newtheorem{example}[thm]{Example}

\def\vect{\mathfrak{vect}}
\def\tr{\mathrm{tr}}
\def\RR{\mathbb{R}}

\title{Feynman path integrals and Lebesgue-Feynman measures}
\author{James Montaldi \& Oleg G.\ Smolyanov}
\date{}

\begin{document}

\maketitle

We call a \emph{Lebesgue-Feynman measure} (LFM) any generalized measure (distribution in the sense of Sobolev and Schwartz) on a locally convex topological vector space (LCS) $E$ which is translation invariant. In what follows $E$ is usually a Hilbert space.  If $\dim E<\infty$ then any LFM on $E$ coincides with (or is defined by) the standard Lebesgue measure (up to a constant real multiple). On the other hand if $E$ is of infinite dimension then by the famous theorem of A.~Weil, there does not exist on $E$ a nonzero $\sigma$-additive locally finite $\sigma$-finite Borel measure which is translation invariant. Clearly, the Lebesgue-Feynman measures does not satisfy these properties of Borel measures.

The idea of a LFM on a space of functions of a real variable, taking values in the configuration space, or  phase space, of a classical Lagrangian, or Hamiltonian, system was used in the very first papers by Feynman; indeed, one can say that the path integrals introduced in those papers give the values taken by the LFM on the integrand. That is, roughly speaking, the integration in the Feynman path integral is integration with respect to the LFM\footnote{It is worth mentioning that Feynman introduced his integration just three years after Weil proved his non-existence theorem}.

In the present paper, we define the class of LFMs on any LCS and investigate transformations of the LFM generated by transformations of the domain and also discuss the connections of these transformations of the LFM with so-called quantum anomalies, improving some recent results \cite{MS,GRS,GMS}.  We revisit the contradiction between the points of view on quantum anomalies presented in the books of Fujikawa and Suzuki \cite{FS} on the one hand, and of Cartier and DeWitt-Morette \cite{CdWM} on the other.

The results of the present paper can be considered as a new formalization of Feynman's results, following in some sense as close as possible Feynman's own approach. It is worth noting that in numerous existing formalizations of Feynman path integrals, one considers only analogues of the Gaussian measure and their analytic continuations, whereas Feynman did not consider any such connection between his path integrals and Gaussian measures.

The connection between the two can be described, in our terminology, as follows. One extracts from the integrand of the Feynman integral, the exponent of $i=\sqrt{-1}$ times the quadratic part of the action (written in either Lagrangian or Hamiltonian form); the product of this exponent and the LFM is just what one calls the Feynman (pseudo-)measure. This means that the latter measure is the generalized measure whose generalized density is this exponent. In general the generalized densities of any measure \cite{SW,Kirillov,MS} can be considered as densities with respect to the LFM.

The paper is organized as follows. In the first section we define the notion of LFM and consider the typical example, the LFM that defines the Feynman path integral. In the second section we consider two approaches to transformations of the LFM: firstly via derivatives along proper vector fields, and secondly via finite dimensional approximations. In the final section we consider some relations between the results obtained in Section 2 and quantum anomalies.  In some places the presentation is more algebraic and precise analytic hypotheses are not elaborated.

\section{Generalized measures on locally convex spaces}

Let $E,G$ be locally convex topological vector spaces. Recall that a mapping $f:E\to G$ is said to be differentiable at $x\in E$ if there exists a continuous linear map $f'(x):E\to G$ (called the derivative at $x$) such that for any convergent sequence $(h_n)$ in $E$ and any sequence $(t_n)$ converging to 0 in $\mathbb{R}$, $f'(x)$ satisfies
$$\lim_{n\to \infty}\left( \frac{f(x+t_nh_n)-f(x)}{t_n} - f'(x)h_n \right) = 0.$$
The derivatives of higher order are defined by induction, assuming that the subsequent spaces in which the derivatives take values are endowed with the topology of compact convergence.

Let $S(E)$ be a real LCS consisting of some class of (continuous) infinitely differentiable complex-values functions on $E$, and denote by $S'(E)$ the dual space equipped with a locally convex topology compatible with the natural duality between $S(E)$ and $S'(E)$. 

\begin{remark} 
The elements of $S'(E)$ can be oconsidered as generalized measures; for any $\nu\in S'(E)$ and any $\varphi\in S(E)$, the value $\nu(\varphi)\equiv(\nu,\varphi)$ will also be denoted $\int\varphi(x)\nu(dx)$ and be called the integral of $\varphi$ with respect to $\nu$. 
\end{remark}

\begin{remark} 
While in the classical theory of distributions, it is often convenient to choose the space $S(E)$ to be `small', in the application to path integrals it is convenient to include in this space the collection of functions that are integrable with respect to generalized measures (of course, when $S(E)$ becomes larger, the dual space $S'(E)$ may become smaller).
\end{remark}

We assume that $S(E)$ satisfies some natural assumptions. In particular, 
\begin{enumerate}
\item For any $h\in E$ the mapping $\varphi\mapsto \varphi'h$ maps $S(E)$ into itself and is continuous;
\item for any $\psi\in S(E)$ and any $h\in E$, the function $\psi_h$ on $E$ defined by 
$$\psi_h(x) = \psi(x+h)$$
(the $h$-shift of $\psi$) also belongs to $S(E)$;
\item the mapping $F_\psi:S(E)\to S(E),\; \psi\mapsto\psi_h$ is continuous,
\item for any $\psi\in S(E)$ the mapping $f_\psi:E\to S(E)$, $h\mapsto\psi_h$ is continuous and differentiable,
\item for any $\psi\in S(E)$ and any $h\in E$ the following identity holds:
$$(f_\psi')(0)h = \psi'(\cdot)h;$$
 consequently $(f_\psi')(0)=\psi$.
\end{enumerate}

\begin{definition}\label{def:derivative1}
Let $h\in E$ and $\nu\in S'(E)$. Then the \emph{derivative of} $\nu$ along $h$ is the element $\nu'h\in S'(E)$ for which, for all $\varphi\in S(E)$, $(\nu'h,\varphi) = -(\nu,\varphi'h)$
\end{definition}

\begin{definition}\label{def:derivative2}
For any $h\in E$, let $F^*_h:S'(E)\to S'(E)$ be the adjoint of $F_h$, and for any $\nu\in S'(E)$ let $\nu_h=F_h^*\nu$ (the $h$-shift of $\nu$) and finally for any $\nu\in S'(E)$ define $\bar{f}_\nu:E\to S'(E)$ by $\bar{f}_\nu(h)=\nu_h$. Then the \emph{derivative of} $\nu$ along $h$ is the element $\nu'h$ of $S'(E)$ defined by $\nu'h=-\bar{f}_\nu'(0)h$.
\end{definition}

\begin{proposition}
For any $\nu\in S'(E)$ and $h\in E$, the two definitions of $\nu'h$ above are equivalent.
\end{proposition}

\begin{proposition}
A generalized measure $\nu$ is translation invariant (i.e., $F_h^*\nu=\nu$ for all $h\in E$) if and only if $\nu'h=0$ for all $h\in E$.
\end{proposition}

We call such translation-invariant measures \emph{Lebesgue-Feynman measures} (or LFMs).

A vector field on $E$ is a mapping $k:E\to E$. Let $\vect(E)$ denote a collection of some infinitely differentiable vector fields on $E$.  If $k\in \vect(E)$ and $\psi\in S(E)$ then the derivative of $\psi$ along $k$ is denoted by $\psi'k$ and is defined by
$$(\psi'k)(x) = \psi'(x)k(x).$$
The derivative of $\nu\in S'(E)$ along a vector field is defined in a similar way.

One extends the $h$-shift of $\psi\in S(E)$ defined above to $k$-shift for $k$ a vector field in the obvious way: $\psi_k(x)=\psi(x+k(x))$, and the mapping $G_k:S(E)\to S(E)$ is defined similarly to $F_h$ by $G_k(\psi)=\psi_k$. 

\begin{remark}
For $\psi\in S(E)$ let $g_\psi:k\mapsto\psi_k$ be the mapping from $\vect(E)$ into $S(E)$ and suppose this mapping is differentiable.  Then
$$(g_\psi')(0)k=\psi'(\cdot)k(\cdot) = (\psi'k)(\cdot).$$
\end{remark}

The analogues of Definitions \ref{def:derivative1} and \ref{def:derivative2} can be formulated also for generalized measures along vector fields.  The analogue of Definition  \ref{def:derivative1} is obtained simply by replacing $h\in E$ by $k\in \vect(E)$. The analogue of Definition \ref{def:derivative2} can be formulated as follows:

\begin{definition}\label{def:derivative2-vfd}
For $k\in \vect(E)$ let $G_k^*$ be the adjoint of $G_k$, and for $\nu\in S'(E)$ let $\nu_k=G_k^*\nu$ (the $k$-shift of $\nu$) and put $g_\nu(t) = \nu_{tk}$ for $t\in\mathbb{R}$. The \emph{the derivative} $\nu'k$ of $\nu$ along $k$ is defined to be
$$\nu'k = -g_\nu'(0).$$
\end{definition}

\begin{proposition}
For any $\nu\in S'(E)$ and $k\in \vect(E)$, the two definitions of $\nu'k$ above are equivalent.
\end{proposition}

The fundamental difference between derivatives of functions and of (generalized) measures is that if a function is translation invariant then its derivative along any vector field vanishes, while the same is not true of measures which are translation invariant, and this fact plays a central role in analysing quantum anomalies: a translation invariant measure ($G_h^*\nu=\nu$ for all $h\in E$) may have non vanishing derivative along a vector field.

In what follows we assume $E$ is a separable Hilbert space; the extension to more general LCSs is also possible. 

We also use a generalized measure taking values in $E$: such an object can either be defined as a continuous linear mapping $S(E)\to E$ or as an element of the topological tensor product of $E$ and $S(E)$, or as a continuous linear functional on the topological tensor product of $E$ with $S(E)$.

\begin{theorem} \label{thm:1}
Let $k\in\vect(E)$ be such that for all $x\in E$, $k'(x)$ exists and is a trace class operator on $E$, and let $\nu\in S'(E)$ be translation invariant (i.e., a Lebesgue-Feynman measure).  Then $\nu'k = \tr(k')\nu$.
\end{theorem}

\begin{proof}
Let $(e_n)$ be an orthonormal basis of $E$. For any $n$ and any $\varphi\in S(E)$, the derivative of the generalized $E$-valued measure $k.\nu$ along $e_n$ satisfies the following identity (we use the translation invariance of $\nu$ and Leibnitz' rule):
$$((k.\nu),\varphi'e_n) = -((k.\nu)'e_n,\varphi) = -((k'e_n).\nu,\varphi).$$
But $\sum_{n}((k,e_n)\nu,(\varphi'e_n)) = (\nu,\varphi' \sum_n(k,e_n)e_n) = (\nu,\varphi'k)$. Hence,
$$(\nu,\varphi'k) = -(\sum_n(k'e_n,e_n)\nu,\varphi) = -(\tr(k')\nu,\varphi).$$ 
This means that $\nu'k=\tr(k')\nu$, as required. 
\end{proof}

\begin{remark} \label{rmk:normalization}
Of course, like in the finite dimensional setting, the Lebesgue-Feynman measure is not unique, but it can be normalized by the assumption that the value it takes on the canaonical Gaussian exponent on $E$ given by $\varphi(x) = e^{-(x,x)/2}$ is equal to 1. This  implies (subject to some analytic details) that for any $\psi\in S(E)$
$$(\nu,\psi) = \lim_{n\to\infty} \frac1{(\sqrt{2\pi})^n}\int_{E_n}\psi(x_1,\dots,x_n) \;dx_1\dots dx_n,$$
where $E_n$ is the vector subspace of $E$ spanned by $\{e_1,\dots,e_n\}$. We will not discuss the question of the extent to which this normalization defines $\nu_{\mathrm{LF}}$ uniquely.
\end{remark}

\begin{remark}
There is another way of defining the Lebesgue-Feynman measure---as a limit of a sequence of $\sigma$-additive Gaussian measures whose correlation operators converge to the identity; again the question of uniqueness remains to be discussed.
\end{remark}

\begin{example}
Let $E=W^1_2((0,t),Q)$, with $\dim Q<\infty$, be the Hilbert space of all absolutely continuous functions $f$ from $[0,t)$ to $Q$ vanishing at $0$ and for which $f'\in L^2((0,t),Q)$. The Hilbert norm we take is $\|f\|=\|\dot f\|_{L^2}$. Let $\nu_Q$ be the Lebesgue-Feynman measure on $E$ normalized as in Remark~\ref{rmk:normalization}. If the function
$$\psi_Q:\xi\mapsto e^{\left(\frac12\int_0^t \dot\xi^2(\tau)d\tau + \int_0^t(V(\xi(\tau)+q)d\tau\right)}\varphi_0(\xi(t)+q)$$
belongs to $S(E)$, then
$$\int_E\psi_q(\xi)\nu_Q(d\xi) = \varphi(t,q),$$
where $\varphi$ is the solution of the Cauchy problem for the Schr\"odinger equation 
$$i\dot\varphi(t,\xi) = -\frac12\Delta_q\varphi(t,q) + V(q)\varphi(t,q)$$
with initial data $\varphi_0$. 
Here we use the first of the two definitions of LFM.

This representation of the solution of the equation is a rigorous version of a formula from the paper of Feynman; we would like to call it the second Feynman formula. The first Feynman formula is a representation of the same solution as a limit of integrals over finite Cartesian powers of the configuration space (or some other space). This limit can be considered as a definition of the Lebesgue-Feynman generalized measure, and the integral over Cartesian products as approximations of the value that the LFM takes on the integrand.
\end{example}

\begin{example}
   Let $E$ be the space of functions defined on $[0,t)$ taking values in the phase space $Q\times P$, with $\dim P=\dim Q<\infty$ and such that if $f\in E$ with $f(\tau)=(q(\tau),p(\tau))$ ($\tau\in[0,t)$) then $q\in L^2((0,t),Q)$ and $p\in W^1_2((0,t),P)$ and 
   $$\|f\|^2=\|\dot q\|^2_{L^2}+\|p\|^2_{L^2}.$$
Let $\mathcal{H}$ be a function (having suitable properties) on $Q\times P$, the classical Hamiltonian function, and $\widehat{\mathcal{H}}$  be the operator on $L^2(Q)$, whose Weyl symbol is $\mathcal{H}$, and suppose the function 
$$\psi_{Q\times P}:(q(\cdot),p(\cdot)) \mapsto 
e^{\left(\frac12\int_0^t (q(\tau)+q)\dot p(\tau) d\tau + \int_0^t(\mathcal{H}(q(\tau)+q,p(\tau))d\tau\right)}\varphi_0(q(t)+q)$$
belongs to $S(E)$ and $\nu_{Q\times P}$ is the LFM on $E$ (again, normlized as in Remark \ref{rmk:normalization}). Then 
$$\int\psi_{Q\times P}(\xi)\nu_{Q\times P}(d\xi) = \varphi(t,q)$$
where $\varphi$ is the solution of the Cauchy problem for the Schr\"odinger equation
$$i\dot\varphi_t = \hat{\mathcal{H}}\varphi(t,\cdot)$$
with initial data $\varphi_0$. This representation gives a rigorous version of another `second Feynman formula' (from another paper by Feynman \cite{Feyn51}). 
\end{example}

\begin{remark}
When the phase space is of infinite dimension then the Schr\"odinger type equations are not so easy to define and it is reasonable to use similar formulas as a natural way of quantization.	
\end{remark}

\section{Transformations of generalized measures}
Using some ideas from \cite{SW} (see also \cite{Steb}) we develop formulas for transformations of the Lebesgue-Feynman measures.

Let $I=(\alpha,\beta)\subset\RR$ and let $F:I\times E\to E$ be a twice continuously differentiable mapping such that $F(0,x)=x$ (for all $x$) and that $\forall t\in I$ the map $F(t,\cdot):x\mapsto F(t,x)$ is infinitely differentiable with infinitely differentiable inverse denoted $F^{-1}(t,\cdot)$. We also assume the mapping $F^{-1}:(t,x)\mapsto F^{-1}(t,x)$ is twice continuously differentiable.

For each $\nu\in S'(E)$ let $f^\nu:I\to S(E)$ be defined by $f^\nu(t) = (F^{-1}(t,\cdot))_*\nu$. 

The following theorem is a corollary of Theorem \ref{thm:1} using the vector field defined by $k(F(t,x))=F_1'(t,x)$ (compare with \cite{SW} where the case of usual measures is considered):

\begin{theorem} \label{thm:2}
	Let $\nu\in S'(E)$ be a LFM. Then the differential of $f^\nu$ at $x$ is given by 
	$$(f^\nu)'(t) = \tr(F''_{12}(t,x)\circ(F_2'(t,x)^{-1})f_\nu(t),$$
where the symbols $F_2'$ and $F_{12}''$ denote the corresponding partial derivatives.
\end{theorem}

\begin{theorem} \label{thm:3}
	Again let $\nu\in S'(E)$. Then $f^\nu(t)=\det F_2'(t,\cdot)\nu$.  Here we assume that the function $x\mapsto\det F_2'(t,x)$ is a multiplicator in $S'(E)$ and hence in $S(E)$.
\end{theorem}

\begin{proof}
	The statement of Theorem~\ref{thm:2} means that the function $f^\nu(\cdot)$ is the solution of the differential equation with initial data $f^\nu(0)=\nu$. This in turn means that 
	$$f^\nu(t) = e^{\int_0^t \tr(F_{12}''(\tau,\cdot)\circ F_2'(\tau,\cdot))^{-1}d\tau}\nu.$$ 
However, $F_{12}''=F_{21}''$ and hence the exponent can be transformed as follows	
\begin{eqnarray*}
&&\exp\left(\int_0^t\tr\left(\frac{d}{d\tau}F_2'(\tau,\cdot)\circ F_2'(\tau,\cdot)\right)^{-1}d\tau\right) \\
&=&	\exp\left(\int_0^t\tr\left(\frac{d}{d\tau}\ln F_2'(\tau,\cdot)\right)d\tau\right) \\
&=& \exp\left(\tr\ln F_2'(t,\cdot)\right) = \det F_2'(t,\cdot),
\end{eqnarray*}
which proves the theorem.
\end{proof}

\begin{theorem}[Change of variables formula for Feynman path integrals]\label{thm:4}
	Let $\varphi\in S(E)$ and $\nu\in S'(E)$. Then, using the integral notation introduced above, 
$$\int_E\varphi(F^{-1}(t,x))\nu(dx) = \int_E\varphi(x)\det(F_2'(t,x))\nu(dx).$$
\end{theorem}

\begin{remark}
The formula of Theorem~\ref{thm:3} can also be obtained, as has already been mentioned, using the definition of the LFM as a limit of integrals over Cartesian products. Namely, the value the generalized measure $\det(F_2'(t,\cdot)\nu$ takes on $\varphi$ can be calculated as a limit of integrals of Cartesian powers of a suitable space; then to each of these integrals we can apply the usual (finite dimensional) change of variables formula and in this way obtain the values that the generalized measure $f_\nu(t)$ takes at the same function.
\end{remark}

\section{Some remarks on quantum anomalies}


If the dependence of the determinant on time and position is such that $\det F_2'(t,x)\neq1$ then the quantum dynamics may not be invariant with respect to this transformation, even though the classical action, and hence the exponential of $\sqrt{-1}$ times this action,  are invariant. 

The remark in the book \cite{CdWM} about the possibility of multiplying this determinant by a function in order to compensate the change in the determinant and in this way to disprove the main claim in \cite{FS} about the origin of quantum anomalies, cannot be considered as correct, because  the quantum dynamics is described by the generalized measure which is the product of the complex exponential of the action and the Lebesgue-Feynman measure times this determinant, and the non-invariance of this dynamics with respect to a transformation cannot disappear. 

\begin{remark}
Consider the statement of Theorem\,\ref{thm:4} above. If the integrand $\varphi$ (which contains the complex exponential of the action) is invariant with respect to some transformation then the (conventional) integrals from Theorem~\ref{thm:4}, which are in fact the values taken by the original and transformed LFMs on the integrand $\varphi$, are also invariant with respect to the given transformation. But some aspects of the (quantum) dynamics may depend on the whole product $\phi.\nu$ which may not be invariant.
\end{remark}

\begin{remark}
There still remains the question of clarifying the role of renormalization in the problem of quantum anomalies.
\end{remark}

This paper was written during a visit of O.G.~Smolyanov to the University of Manchester, which was partially funded by a Scheme 2 grant from the London Mathematical Society.  He thanks the University for its hospitality and excellent conditions for scientific work.


JM: School of Mathematics, University of Manchester, Manchester M13 9PL, UK.

OGS: Lomonosov Moscow State University, Moscow, Russia


\small
\begin{thebibliography}{99}

\bibitem{CM}  R.H.~Cameron \& W.T.~Martin, Transformations of Wiener integrals under translations. \textit{Annals of Math.} \textbf{45}  (1944), 386--396.
 
 \bibitem{CdWM} P.~Cartier \& C~DeWitt-Morette, \emph{Functional Integration: Action and Symmetries}. CUP, 2006
 
\bibitem{Feyn48} R.~Feynman, Space-time approach to non-relativistic quantum mechanics. 
\textit{Rev.\ Mod.\ Phys.} \textbf{20} (1948), 367--387.

\bibitem{Feyn51} R.~Feynman, An operator calculus having applications in quantum electrodynamics. \textit{Phys.\ Rev.} \textbf{84} (1951), 108--128.

\bibitem{FS} K.~Fujikawa \& H.~Suzuki, \emph{Path Integrals and Quantum Anomalies}. Oxford, 2004.


\bibitem{GMS} L.C.~Garc\'ia-Naranjo, J.~Montaldi \& O.G.~Smolyanov, Transformations of Feynman path integrals and generalized
densities of Feynman pseudomeasures. \textit{Dokl.\ Akad.\ Nauk} \textbf{468} (2016), 367--371 (Russian); translation in  \textit{Dokl.\ Math.} \textbf{93} (2016), 1--4.
  
 \bibitem{GY1956} I.M.~Gelfand \& A.M.~Yaglom, Integration in function spaces and its application to quantum physics. (Russian) \textit{Uspehi Mat.\ Nauk} (N.S.) \textbf{11} (1956), 77--114.

\bibitem{GRS} J.~Gough, T.S.~Ratiu \& O.G.~Smolyanov, Quantum anomalies and logarithmic derivatives of Feynman pseudomeasures.  \textit{Dokl.\ Akad.\ Nauk} \textbf{465} (2015), 651--655 (Russian); translation in \textit{Dokl.\ Math.} \textbf{92} (2015), 764--768.

\bibitem{Kirillov} A.I.~Kirillov, Infinite-dimensional analysis and quantum theory as semimartingale calculus. \textit{Russian Math.\ Surveys} \textbf{49} (1994), 43--95.

\bibitem{MS} J.~Montaldi \& O.G.~Smolyanov,  Transformations of measures via their generalized densities. 
\textit{Russ.\ J.\ Math.\ Phys.} \textbf{21} (2014), 379--385.

\bibitem{SW} O.G.\ Smolyanov \& H.v.~Weizs\"acker, Change of measures and their logarithmic derivatives under smooth transformations.  \textit{C.R.\ Acad.\ Sci., Paris} \textbf{321} (1995), 103--108.


\bibitem{Steb} V.R.~Steblovskaya, in  Yu.L.~Dalecky, S.V.~Fomin, \textit{Measures and differential equations in infinite-dimensional space}. Translated from the Russian. Mathematics and its Applications (Soviet Series), Vol.~76, 1992. 

\end{thebibliography}
\end{document}